\documentclass[11pt,letterpaper]{article}

\usepackage{amsmath,amsfonts,amsthm,amssymb,stmaryrd,relsize}
\usepackage{cite} % added 9.18.20 to collect multiple sequential cites. [1-5] 
\theoremstyle{plain}

\numberwithin{equation}{section}

\newtheorem{thm}{Theorem}[section]
\newtheorem{lem}[thm]{Lemma}
\newtheorem{cor}[thm]{Corollary}

\newenvironment{exam}[1]
{\begin{flushleft}\textbf{Example #1}.\enspace}%
{\end{flushleft}}

\allowdisplaybreaks  % introduced 7.15.15

\newcommand{\complex}{{\mathbb C}}
\newcommand{\real}{{\mathbb R}}

\newcommand{\tbullet}{\raise .4ex\hbox{\tiny$\bullet$}} % 5.8.20

\newcommand{\rmtr}{\mathrm{tr\,}}
\newcommand{\rmin}{\mathrm{In\,}}

\newcommand{\ascript}{\mathcal{A}}
\newcommand{\cscript}{\mathcal{C}}
\newcommand{\escript}{\mathcal{E}}
\newcommand{\iscript}{\mathcal{I}}
\newcommand{\lscript}{\mathcal{L}}
\newcommand{\mscript}{\mathcal{M}}
\newcommand{\oscript}{\mathcal{O}}
\newcommand{\pscript}{\mathcal{P}}
\newcommand{\sscript}{\mathcal{S}}
\newcommand{\mscripthat}{\widehat{\mscript}}
\newcommand{\mscripttilde}{\widetilde{\mscript}}
\newcommand{\iscripthat}{\widehat{\iscript}}

\newcommand{\ab}[1]{\left|#1\right|}
\newcommand{\doubleab}[1]{\left|\left|#1\right|\right|}
\newcommand{\brac}[1]{\left\{#1\right\}}
\newcommand{\paren}[1]{\left(#1\right)}
\newcommand{\sqbrac}[1]{\left[#1\right]}
\newcommand{\elbows}[1]{{\left\langle#1\right\rangle}}
\newcommand{\ket}[1]{{\left|#1\right>}}
\newcommand{\bra}[1]{{\left<#1\right|}}

\begin{document}

\title{NONDISTURBING QUANTUM\\MEASUREMENT MODELS}
\author{Stan Gudder\\ Department of Mathematics\\
University of Denver\\ Denver, Colorado 80208\\
sgudder@du.edu}
\date{}
\maketitle

\begin{abstract}
A measurement model is a framework that describes a quantum measurement process. In this article we restrict attention to $MM$s on finite-dimensional Hilbert spaces. Suppose we want to measure an observable $A$ whose outcomes $A_x$ are represented by positive operators (effects) on a Hilbert Space $H$. We call $H$ the base or object system. We interact $H$ with a probe system on another Hilbert space $K$ by means of a quantum channel. The probe system contains a probe (or meter or pointer) observable $F$ whose outcomes $F_x$ are measured by an apparatus that is frequently (but need not be) classical in practice. The $MM$ protocol gives a method for determining the probability of an outcome $A_x$ for any state of $H$ in terms of the outcome $F_x$. The interaction channel usually entangles this state with an initial probe state of $K$ that can be quite complicated. However, if the channel is nondisturbing in a sense that we describe, then the entanglement is considerably simplified. In this article, we give formulas for observables and instruments measured by nondisturbing $MM$s. We begin with a general discussion of nondisturbing operators relative to a quantum context. We present two examples that illustrate this theory in terms of unitary nondisturbing channels.
\end{abstract}

\section{Introduction}  % Section 1
This section discusses the basic concepts and definitions that are needed in the sequel. For more details and motivation we refer the reader to \cite{bgl95,gud120,gud220,hz12,kra83,nc00}. We shall only consider finite-dimensional Hilbert spaces $H$ and $K$. Let $\lscript (H)$ be the set of linear operators on $H$. For $S,T\in\lscript (H)$ we write $S\le T$ if $\elbows{\phi ,S\phi}\le\elbows{\phi ,T\phi}$ for all $\phi\in H$. We define the set of \textit{effects} on $H$ by
\begin{equation*}
\escript (H)=\brac{a\in\lscript (H)\colon 0\le a\le I}
\end{equation*}
where $0,I$ are the zero and identity operators, respectively. Effects correspond to yes-no measurements and when the result of a measuring
$a$ is yes, we say that $a$ \textit{occurs}. A one-dimensional projection $P_\phi =\ket{\phi}\bra{\phi}$ where $\doubleab{\phi}=1$ is an effect called an \textit{atom}. We call $\rho\in\escript (H)$ a \textit{partial state} if $\rmtr (\rho )\le 1$ and $\rho$ is a \textit{state} if $\rmtr (\rho )=1$. We denote the set of partial states by $\sscript _p(H)$ and the set of states by $\sscript (H)$. If $\rho\in\sscript (H)$, $a\in\escript (H)$, we call
$\pscript _\rho (a)=\rmtr (\rho a)$ the \textit{probability} that $a$ \textit{occurs in the state} $\rho$ \cite{bgl95,hz12,nc00}.

Let $\Omega _A$ be a finite set. A (finite) \textit{observable with outcome-space} $\Omega _A$ is a subset
\begin{equation*}
A=\brac{A_x\colon x\in\Omega _A}\subseteq\escript (H)
\end{equation*}
that satisfies $\sum\limits _{x\in\Omega _A}A_x=I$. We interpret $A_x$ as the effect that occurs when $A$ has outcome $x$. We denote the set of observables on $H$ by $\oscript (H)$. If $A\in\oscript (H)$, we define the \textit{effect-valued measure} $X\to A_X$ from $2^{\Omega _A}$ to
$\escript (H)$ by $A_X=\sum\limits _{x\in X}A_x$. We interpret $A_X$ as the event that $A$ has an outcome in $X$. If $\rho\in\sscript (H)$ and $A\in\oscript (H)$, the \textit{probability that} $A$ \textit{has an outcome} in $X\in\Omega _A$ when the system is in state $\rho$ is
$\pscript _\rho (A_X)=\rmtr (\rho A_X)$. Notice that $X\mapsto\pscript _\rho (A_X)$ is a probability measure on $\Omega _A$ \cite{bgl95,hz12}.

An \textit{operation} is a completely positive map $\ascript\colon\sscript _p(H)\to\sscript _p(H)$ \cite{hz12,nc00}. Every operation has a
\textit{Kraus decomposition} \cite{hz12,kra83,nc00}
\begin{equation*}
\ascript (\rho )=\sum _{i=1}^nS_i\rho S_i^*
\end{equation*}
where $S_i\in\lscript (H)$ with $\sum\limits _{i=1}^nS_i^*S_i\le I$. An operation $\ascript$ is a \textit{channel} if $\ascript (\rho )\in\sscript (H)$ for every $\rho\in\sscript (H)$. In this case, $\sum\limits _{i=1}^nS_i^*S_i=I$ and we denote the set of channels by $\cscript (H)$. For a finite set
$\Omega _\iscript$, an \textit{instrument with outcome-set} $\Omega _\iscript$ is a set of operations
$\iscript =\brac{\iscript _x\colon\in\Omega _\iscript}$ satisfying \cite{bgl95,gud120,gud220,hz12}
\begin{equation*}
\cscript _\iscript =\sum\brac{\iscript _x\colon\Omega _\iscript}\in\cscript (H)
\end{equation*}
Defining $\iscript _X=\sum\limits _{x\in X}\iscript _x$ for $X\subseteq\Omega _\iscript$ we see that $X\mapsto\iscript _X$ is an
\textit{operation-valued measure} on $H$. We denote the set of instruments on $H$ by $\rmin (H)$. We say that $\iscript\in\iscript (H)$
\textit{measures} $A\in\oscript (H)$ if $\Omega _\iscript =\Omega _A$ and
\begin{equation*}
\pscript _\rho (Ax)=\rmtr\sqbrac{\iscript _x(\rho )}
\end{equation*}
for all $\rho\in\sscript (H)$, $x\in\Omega _A$. There is a unique $A\in\oscript (H)$ that $\iscript$ measures and we write $A=\iscripthat$
\cite{gud120,gud220,hz12}.

A \textit{measurement model} ($MM$) is a 5-tuple $\mscript =(H,K,\eta ,\nu ,F)$ where $H,K$ are finite-dimensional Hilbert spaces called the
\textit{base} and \textit{probe} systems, respectively, $\eta\in\sscript (K)$ is an \textit{initial probe state}, $\nu\in\cscript (H\otimes K)$ is a channel describing the measurement interaction between the base and probe systems and $F\in\oscript (K)$ is the \textit{probe} (or \textit{meter})
\textit{observable} \cite{bgl95,gud120,gud220,hz12}. We say that $\mscript$ \textit{measures the model instrument} $\mscripthat\in\rmin (H)$ where $\mscripthat$ is the unique instrument satisfying
\begin{equation}                % equation (1.1)
\label{eq11}
\mscripthat _x(\rho )=\rmtr _K\sqbrac{\nu (\rho\otimes\eta )(I\otimes F_x)}
\end{equation}
for all $\rho\in\sscript (H)$, $x\in\Omega _F$. In \eqref{eq11}, $\rmtr _K$ is the partial trace over $K$ \cite{bgl95,hz12,nc00}. We also say that
$\mscript$ \textit{measures the model} observable $\mscript^{\wedge\wedge}$.

We thus have three levels of abstraction. At the basic level is an observable that we seek to measure. At the next level is an instrument
$\iscript$ which is an apparatus that can be employed to measure an observable $\iscripthat$. Although $\iscripthat$ is unique, there are many instruments that can be used to measure an observable. Moreover, $\iscript$ gives more information that $\iscripthat$ because, depending on the outcome $x$ (or event $X$), $\iscript$ updates the input state $\rho$ to give the output partial state $\iscript _x(\rho )$ (or $\iscript _X(\rho )$). At the highest level is a measurement model $\mscript$ that measures a unique model instrument $\mscripthat$ and unique observable
$\mscript ^{\wedge\wedge}$. Again, there are many $MM$s that measure any instrument or observable and $\mscript$ contains more information on how the measurement is performed.

\section{Nondisturbing Operators}  % Section 2
Let $H$ and $K$ be finite dimensional complex Hilbert spaces for the base and probe systems, respectively. Let $\brac{\psi _i}$,
$i=1,2,\ldots ,n$, be an orthonormal basis for $H$ and let $C=\brac{P_{\psi _i}}$ be the corresponding atomic observables on $H$. We call $C$ a \textit{context} for $H$ and think of $C$ as a particular way of viewing the base system. Of course, there are many contexts and each provides a different view of $H$. If $S\in\lscript (H)$ has the form $S=\sum c_iP_{\psi _i}$, $c_i\in\complex$, we say that $S$ is \textit{measurable} with respect to $C$. In particular, any self-adjoint operator that commutes with $P_{\psi _i}$ for all $i=1,2,\ldots ,n$ is measurable with respect to $C$. Letting $I_K$ be the identity operator on $K$, an operator $A\in\lscript (H\otimes K)$ is $C$-\textit{nondisturbing} if
\begin{equation*}
A(P_{\psi _i}\otimes I_K)=(P_{\psi _i}\otimes I_K)A
\end{equation*}
$i=1,2,\ldots ,n$. We think of $C$-nondisturbing operators as those operators on $H\otimes K$ that leave the context invariant. For example, if $A=D\otimes E$ where $D$ is measurable with respect to $C$, then $A$ is $C$-nondisturbing. Of course, if $A$ is $C$-nondisturbing, $A$ may not be $C'$-nondisturbing for a different context $C'$. The $C$-nondisturbing operators form a $C^*$-subalgebra of $\lscript (H\otimes K)$. 

\begin{thm}    % Theorem 2.1
\label{thm21}
{\rm{(a)}}\enspace The following statements are equivalent:
{\rm{(i)}}\enspace $A$ is $C$-nondisturbing.
{\rm{(ii)}}\enspace There exist operators $B_i\in\lscript (K)$, $i=1,2,\ldots ,n$, such that $A(\psi _i\otimes\phi )=\psi _i\otimes B_i\phi$ for all
$\phi\in K$.
{\rm{(iii)}}\enspace There exist operators $B_i\in\lscript (K)$, $i=1,2,\ldots ,n$ such that $A=\sum\limits _{i=1}^n(P_{\psi _i}\otimes B_i)$.
{\rm{(b)}}\enspace The operators in {\rm{(ii)}} and {\rm{(iii)}} are unique and satisfy
\begin{equation}                % equation (2.1)
\label{eq21}
B_i\phi=\sum _j\elbows{\psi _i\otimes\phi _j,A(\psi _i\otimes\phi )}\phi _j
\end{equation}
for every $\phi\in K$ and every orthonormal basis $\phi _j$ for $K$.
\end{thm}
\begin{proof}
(a)\enspace To show that (i) implies (ii) suppose that $A$ is $C$-nondisturbing. Let $\brac{\phi _j}$ be an orthonormal basis for $K$ and define
$B_i\in\lscript (H)$, $i=1,2,\ldots ,n$ by \eqref{eq21}. We then have that
\begin{align*}
A(\psi _i\otimes\phi )&=A(P_{\psi _i}\otimes I_K)(\psi _i\otimes\phi )=(P_{\psi _i}\otimes I_K)A(\phi _i\otimes\phi )\\
  &=\sum _{r,s}\elbows{\psi _r\otimes\phi _s,(P_{\psi _i}\otimes I_K)A(\psi _i\otimes\phi )}\psi _r\otimes\phi _s\\
  &=\sum _s\elbows{\psi _i\otimes\phi _s,A(\psi _i\otimes\phi )}\psi _i\otimes\phi _s=\psi _i\otimes B_i\phi
\end{align*}
Hence, (ii) holds. To show that (ii) implies (iii) suppose (ii) holds. We conclude that
\begin{equation*}
\sum _{i=1}^n(P_{\psi _j}\otimes B_i)(\psi _j\otimes\phi )=\psi _j\otimes B_j\phi =A(\psi _j\otimes\phi )
\end{equation*}
for all $j=1,2,\ldots ,n$, $\phi\in K$. Hence, (iii) holds. To show that (iii) implies (i) suppose that (iii) holds. We then obtain
\begin{equation*}
A(P_{\psi _i}\otimes I_K)=P_{\psi _j}\otimes B_j=(P_{\psi _j}\otimes I_K)A
\end{equation*}
so (i) holds.
(b)\enspace If (ii) holds, then we have that
\begin{align*}
B_i\phi&=\sum _j\elbows{\phi _j,B_i\phi}\phi _j=\sum _j\elbows{\psi _i\otimes\phi _j,\psi _i\otimes B_i\phi}\phi _j\\
   &=\sum _j\elbows{\psi _i\otimes\phi _j,A(\psi _i\otimes\phi )}\phi _j
\end{align*}
so \eqref{eq21} holds.
\end{proof}

If $A$ is $C$-nondisturbing, the operators $B_i\in\lscript (K)$, $i=1,2,\ldots ,n$ in Theorem~\ref{thm21} are called the corresponding
\textit{probe operators}.

\begin{cor}    % Corollary 2.2
\label{cor22}
Let $A\in\lscript (H\otimes K)$ be $C$-nondisturbing with probe operators $B_i$, $i=1,2,\ldots ,n$
{\rm{(a)}}\enspace $B_i=\rmtr _H\sqbrac{A(P_{\psi _i}\otimes I_K)}$.
{\rm{(b)}}\enspace $\rmtr _H(A)=\sum B_i$.
{\rm{(c)}}\enspace $\rmtr _K(A)=\sum\rmtr (B_i)P_{\psi _i}$ so $\rmtr _K(A)$ is $C$-measurable.
\end{cor}
\begin{proof}
(a)\enspace If $\brac{\phi _k}$ is an orthonormal basis for $K$ we have by Theorem~\ref{thm21} that
\begin{align*}
B_i&=\sum _{k,l}\elbows{\phi _k,B_i\phi _l}\ket{\phi _k}\bra{\phi _l}
   =\sum _{k,l}\elbows{\psi _i\otimes\phi _k,\psi _i\otimes B_i\phi _l}\ket{\phi _k}\bra{\phi _l}\\
   &=\sum _{k,l}\elbows{\psi _i\otimes\phi _k,A(\psi _i\otimes\phi _l)}\ket{\phi _k}\bra{\phi _l}\\
   &=\sum _{j,k,l}\elbows{\psi _i\otimes\phi _k,A(P_{\psi _i}\otimes I_K)(\psi _j\otimes\phi _l)}\ket{\phi _k}\bra{\phi _l}\\
   &=\rmtr _H\sqbrac{A(P_{\psi _i}\otimes I_K)}
\end{align*}
(b)\enspace Similar to (a) we obtain
\begin{align*}
\sum _iB_i&=\sum _{i,j,k}\elbows{\phi _j,B_i\phi _k}\ket{\phi _j}\bra{\phi _k}
    =\sum _{i,j,k}\elbows{\psi _i\otimes\phi _j,\psi _i\otimes B_i\phi _k}\ket{\phi _j}\bra{\phi _k}\\
    &=\sum _{i,j,k}\elbows{\psi _i\otimes\phi _j,A(\psi _i\otimes\phi _k)}\ket{\phi _j}\bra{\phi _k}=\rmtr _H(A)
\end{align*}
(c)\enspace By the definition of $\rmtr _K$ we obtain
\begin{align*}
\rmtr _K(A)&=\sum _{i,j,k}\elbows{\psi _i\otimes\phi _j,A(\psi _k\otimes\phi _j)}\ket{\psi _i}\bra{\psi _k}\\
   &=\sum _{i,j,k}\elbows{\psi _i\otimes\phi _j,\psi _k\otimes B_k\phi _j}\ket{\psi _i}\bra{\psi _k}\\
   &=\sum _{i,j}\elbows{\phi _j,B_i\phi _j}\ket{\psi _i}\bra{\psi _i}=\sum _i\rmtr (B_i)P_{\psi _i}\qedhere
\end{align*}
\end{proof}

It follows from Corollary~\ref{cor22}(c) that $\rmtr _K(A)$ is unitary, self-adjoint, an effect or a projection if and only if $\ab{\rmtr (B_i)}=1$,
$\rmtr (B_i)\in\real$, $0\le\rmtr (B_i)\le 1$ or $\rmtr (B_i)\in(0,1)$ for all $i=1,2,\ldots ,n$, respectively.

\begin{lem}    % Lemma 2.3
\label{lem23}
Let $A\in\lscript (H\otimes K)$ be $C$-nondisturbing with probe operators $B_i$, $i=1,2,\ldots ,n$.
{\rm{(a)}}\enspace For any $\phi\in K$ and orthonormal basis $\brac{\phi _j}$ for $K$ we have that
\begin{equation*}
B_i^*\phi =\sum _j\elbows{\psi _i\otimes\phi _j,A^*(\psi _i\otimes\phi )}\phi _j
\end{equation*}
{\rm{(b)}}\enspace $A^*(\psi _i\otimes\phi )=\psi _i\otimes B_i^*\phi$.
{\rm{(c)}}\enspace $A$ is self-adjoint or unitary or a projection if and only if $B_i$ are self-adjoint or unitary or projections, $i=1,2,,,\ldots ,n$, respectively.
{\rm{(d)}}\enspace If $C_i\in\lscript (K)$ are the probe operators for a $C$-nondisturbing operator $D\in\lscript (H\otimes K)$, then $A\le D$ if and only if $B_i\le C_i$, $i=1,2,\ldots ,n$.
{\rm{(e)}}\enspace $A\in\escript (H\otimes K)$ if and only if $B_i\in\escript (K)$, $i=1,2,\ldots ,n$.
{\rm{(f)}}\enspace $A\in\escript (H\otimes K)$ with $\rmtr _H(A)=I$ if and only if $\brac{B_i}\in\oscript (K)$.
\end{lem}
\begin{proof}
(a)\enspace Applying \eqref{eq21} we obtain for every $\phi\in K$
\begin{align*}
B_i^*\phi&=\sum _j\elbows{\phi _j,B_i^*\phi}\phi _j=\sum _j\elbows{B_i\phi _j,\phi}\phi _j
  =\sum _j\elbows{\psi _i\otimes B_i\phi _j,\psi _i\otimes\phi}\phi _j\\
  &=\sum _j\elbows{A(\psi _i\otimes\phi _j),\psi _i\otimes\phi}\phi _j=\sum _j\elbows{\psi _i\otimes\phi _j,A^*(\phi _i\otimes\phi )}\phi _j
\end{align*}
(b)\enspace This follows from (a).
(c)\enspace If $B_i$ are self-adjoint, $i=1,2,\ldots ,n$, then by (b) we obtain
\begin{equation*}
A^*(\psi _i\otimes\phi )=\psi _i\otimes B_i^*\phi =\psi _i\otimes B_i\phi =A(\psi _i\otimes\phi )
\end{equation*}
Hence, $A^*=A$ so $A$ is self-adjoint. If $A$ is self-adjoint, the by (b) we have that
\begin{equation*}
\psi _i\otimes B_i\phi =A(\psi _i\otimes\phi )=A^*(\psi _i\otimes\phi )=\psi _i\otimes A^*\phi
\end{equation*}
Hence, $B_i=B_i^*$ so $B_i$ is self-adjoint, $i=1,2,\ldots ,n$. That $A$ is unitary if and only if $B_i$ are unitary, $i=1,2,\ldots ,n$, is similar. Since $A^2=A$ is equivalent to
\begin{equation*}
A^2(\psi _i\otimes\phi )=\psi _i\otimes B_i^2\phi =A(\psi _i\otimes\phi )=\psi _i\otimes B_i\phi
\end{equation*}
which is equivalent to $B_i=B_i^2$, $i=1,2,\ldots ,n$, we conclude that $A$ is a projection if and only if $B_i$ are projections, $i=1,2,\ldots ,n$.
(d)\enspace If $A\le D$ then $A(P_{\psi _i}\otimes I_K)\le D(P_{\psi _i}\otimes I_K)$, $i=1,2,\ldots ,n$. Since $\rmtr _H$ preserves order it follows from Corollary~\ref{cor22}(a) that $B_i\le C_i$, $i=1,2,\ldots ,n$. Conversely, if $B_i\in C_i$, then
$P_{\psi _i}\otimes B_i\le P_{\psi _i}\otimes C_i$, $i=1,2,\ldots ,n$. Applying Theorem~\ref{thm21}(iii) we conclude that $A\le D$.
(e)\enspace This follows from (c) and (d).
(f)\enspace This follows from (e) and Corollary~\ref{cor22}(b).
\end{proof}

\begin{lem}    % Lemma 2.4
\label{lem24}
Let $A\in\lscript (H\otimes K)$ be $C$-nondisturbing with probe operators $B_i$, $i=1,2,\ldots ,n$.
{\rm{(a)}}\enspace For any $B\in\lscript (H)$, $D\in\lscript (K)$ we have that
\begin{equation}                % equation (2.2)
\label{eq22}
A(B\otimes D)A^*=\sum _{i,j}(P_{\psi _i}BP_{\psi _j}\otimes B_iDB_j^*)
\end{equation}
{\rm{(a)}}\enspace For $k=1,2,\ldots ,n$ we obtain
\begin{equation*}
A(P_{\psi _k}\otimes I_K)A^*=P_{\psi _k}\otimes B_kB_k^*
\end{equation*}
\end{lem}
\begin{proof}
(a)\enspace This follows from Theorem~\ref{thm21}(iii).
(b)\enspace Applying (a) we conclude that
\begin{equation*}
A(P_{\psi _k}\otimes I_K)A^*=\sum _{i,j}(P_{\psi _i}P_{\psi _k}P_{\psi _j}\otimes B_iI_KB_j^*)=P_{\psi _k}\otimes B_kB_k^*\qedhere
\end{equation*}
\end{proof}

\section{Nondisturbing Channels}  % Section 3
A channel $\nu\in\cscript (H\otimes K)$ is $C$-\textit{nondisturbing} if $\nu$ has a Kraus decomposition
$\nu (\sigma )=\sum\limits _{k=1}^mS_k\sigma S_k^*$ where $\sum S_k^*S_k=I$ and each $S_k$ is $C$-nondisturbing.

\begin{thm}    % Theorem 3.1
\label{thm31}
A channel $\nu\in\cscript (H\otimes K)$ is $C$-nondisturbing if and only if there exist channels $\Gamma _i\in\cscript (K)$, $i=1,2,\ldots ,n$, where $\Gamma _i(\eta )=\sum _{k=1}^mB_i^k\eta B_i^{k*}$ and we have that
\begin{equation}                % equation (3.1)
\label{eq31}
\nu (\rho\otimes\eta )=\sum _{i,j,k}(P_{\psi _i}\rho P_{\psi _j}\otimes B_i^k\eta B_j^{k*})
\end{equation}
for all $\rho\in\sscript (H)$, $\eta\in\sscript (K)$.
\end{thm}
\begin{proof}
Suppose $\nu\in\cscript (H\otimes K)$ is $C$-nondisturbing. Then $\nu$ has a Kraus decomposition $\nu (\sigma )=\sum S_k\sigma S_k^*$ where each $S_k$ is $C$-nondisturbing. Let $B_i^k$ be the probe operators for $S_k$, $i=1,2,\ldots ,n$, $k=1,2,\ldots ,m$. Applying
Lemma~\ref{lem24} we obtain
\begin{equation*}
\nu (\rho\otimes\eta )=\sum _kS_k(\rho\otimes\eta )S_k^*=\sum _{i,j,k}(P_{\psi _i}\rho P_{\psi _j}\otimes B_i^k\eta B_j^{k*})
\end{equation*}
which gives \eqref{eq31}. Since $\sum S_k^*S_k=I$, we conclude from Lemma~\ref{lem23}(b) that
\begin{align*}
\psi _i\otimes\phi &=\sum _kS_k^*S_k(\psi _i\otimes\phi )=\sum _kS_k^*(\psi _i\otimes B_i^k\phi )
   =\sum _k(\psi _i\otimes B_i^{k*}B_i^k\phi )\\
   &=\psi _i\otimes\sum _kB_i ^{k*}B_i^k\phi
\end{align*}
for all $\phi\in K$. Hence, $\sum\limits _kB_i^{k*}B_i^k=I_K$ so $\Gamma _i(\eta )=\sum\limits _kB_i^k\eta B_i^{k*}$ is a channel,
$i=1,2,\ldots ,n$. Conversely, suppose we have these channels $\Gamma _i\in\cscript (K)$, $i=1,2,\ldots ,n$ and \eqref{eq31} holds. Define the operators $S_k\in\lscript (H\otimes K)$ by $S_k(\psi _i\otimes\phi )=\psi _i\otimes B_i^k\phi$, $i=1,2,\ldots ,n$, $k=1,2,\ldots ,m$. By
Theorem~\ref{thm21}(a), $S_k$ is $C$-nondisturbing, $k=1,2,\ldots ,m$, and by \eqref{eq31} and our previous calculation we obtain
$\nu (\rho\otimes\eta )=\sum S_k(\rho\otimes\eta )S_k^*$. Since any $\sigma\in\sscript (H\otimes K)$ is a linear combination of product states, we conclude that $\nu (\sigma )=\sum S_k\sigma S_k^*$ for all $\sigma\in\sscript (H\otimes K)$.
\end{proof}

It is interesting to note that if $\rho$ is $C$-measurable so that $\rho P_{\psi _i}=\lambda _iP_{\psi _i}$, $\lambda _i\in\real$, then by
\eqref{eq31}
\begin{equation*}
\nu (\rho\otimes\eta )=\sum _i(\lambda _iP_{\psi _i}\otimes\sum _kB_i^k\eta B_i^{k*})
  =\sum _i\sqbrac{\lambda _iP_{\psi _i}\otimes\Gamma _i(\eta )}
\end{equation*}

\begin{cor}    % Corollary 3.2
\label{cor32}
If $\nu\in\cscript (H\otimes K)$ is a $C$-nondisturbing channel, then $\nu$ has the form \eqref{eq31} and
\begin{align}                % equation (3.2)
\label{eq32}
\rmtr _K\sqbrac{\nu (\rho\otimes\eta )}&=\sum _{i,j,k}\rmtr (B_i^k\eta B_j^{k*})P_{\psi _i}\rho P_{\psi _j}\notag\\
\rmtr _H\sqbrac{\nu (\rho\otimes\eta )}&=\sum _j\elbows{\psi _i,\rho\psi _i}\Gamma _i(\eta )
\end{align}
\end{cor}
\begin{proof}
Applying \eqref{eq31} gives
\begin{align*}
\rmtr _K\sqbrac{\nu (\rho\otimes\eta )}&=\sum _{i,j,k}\rmtr _K(P_{\psi _i}\rho P_{\psi _j}\otimes B_i^k\eta B_j^{k*})\\
   &=\sum _{i,j,k}\rmtr (B_i^k\eta B_j^{k*})P_{\psi _i}\rho P_{\psi _j}
\end{align*}
Moreover,
\begin{align*}
\rmtr _H\sqbrac{\nu (\rho\otimes\eta )}&=\sum _{i,j,k}\rmtr _H(P_{\psi _i}\rho P_{\psi _j}\otimes B_i^k\eta B_j^{k*})\\
  &=\sum _{i,j,k}\rmtr (P_{\psi _i}\rho P_{\psi _j})B_i^k\eta B_j^{k*}\\
  &=\sum _i\elbows{\psi _i,\rho\psi _i}\sum _kB_i^k\eta B_i^{k*}=\sum _i\elbows{\psi _i,\rho\psi _i}\Gamma _i(\eta )\qedhere
\end{align*}
\end{proof}

We see from \eqref{eq32} that for every $\rho\in\sscript (H)$ the map $\eta\mapsto\rmtr _H\sqbrac{\nu (\rho\otimes\eta )}$ is a channel on $K$ that is a convex combination of the channels $\Gamma _i$, $i=1,2,\ldots,n$.

We now arrive at our most important definition. We say that a $MM$ $\mscript =(H,K,\eta ,\nu ,F)$ is $C$-\textit{nondisturbing} if the channel
$\nu$ is $C$-nondisturbing.

\begin{thm}    % Theorem 3.3
\label{thm33}
If $\mscript =(H,K,\eta ,\nu ,F)$ is a $C$-nondisturbing $MM$, then $\nu$ is given by \eqref{eq31}. The instrument and observable measured by
$\mscript$ become
\begin{align}                
\label{eq33}                        % equation (3.3)
\mscripthat _x(\rho )&=\sum _{i,j,k}\rmtr (B_i^k\eta B_j^{k*}F_x)P_{\psi _i}\rho P_{\psi _j}\\
\label{eq34}                          % equation (3.4)
\mscript _x^{\wedge\wedge}&=\sum _i\rmtr\sqbrac{\Gamma _i(\eta )F_x}P_{\psi _i}
\end{align}
\end{thm}
\begin{proof}
For any $\rho\in\sscript (H)$, applying \eqref{eq21} and \eqref{eq31} we obtain
\begin{align*}
\mscripthat _x(\rho )&=\rmtr _K\sqbrac{(\rho\otimes\eta )(I\otimes F_x)}
   =\sum _{i,j,k}\rmtr _K(P_{\psi _i}\rho P_{\psi _j}\otimes B_i^k\eta B_j^{k*}F_x)\\
   &=\sum _{i,j,k}\rmtr (B_i^k\eta B_j^{k*}F_x)P_{\psi _i}\rho P_{\psi _j}
\end{align*}
Let $A_x=\sum _i\rmtr\sqbrac{\Gamma _i(\eta )F_x}P_{\psi _i}$ as in \eqref{eq34}. Applying \eqref{eq33} we conclude that
\begin{align*}
\rmtr\sqbrac{\mscripthat _x(\rho )}&=\sum _{i,k}\rmtr (B_i^k\eta B_i^{k*}Fx)\elbows{\psi _i,\rho\psi _i}\\
   &=\sum _i\rmtr\sqbrac{\Gamma _i(\eta )F_x}\elbows{\psi _i,\rho\psi _i}=\rmtr (\rho A_x)
\end{align*}
Hence, $A_x=\mscript _x^{\wedge\wedge}$ which proves \eqref{eq34}.
\end{proof}

Notice that $\mscript _x^{\wedge\wedge}$ is $C$-measurable and the effects $\mscript _x^{\wedge\wedge}$ commute for all $x\in\Omega _F$. This gives a restriction on the measured observable for a $C$-nondisturbing $MM$. We conclude that not all observables can be measured by a $C$-nondisturbing $MM$. In general $\mscripthat _x(\rho )$ need not be $C$-measurable. However, if $\rho$ is $C$-measurable so that $\rho\psi _i=\lambda _i\psi _i$ for $i=1,2,\ldots ,n$, then \eqref{eq33} becomes:
\begin{equation*}
\mscripthat _x(\rho )=\sum _{i,k}\rmtr (B_i^k\eta B_i^{k*}F_x)\lambda _iP_{\psi _i}
  =\sum _i\lambda _i\rmtr\sqbrac{\Gamma _i(\eta )F_x}P_{\psi _i}
\end{equation*}
which is $C$-measurable.

Although $\rmtr$ is cyclic in the sense that $\rmtr (AB)=\rmtr (BA)$, this need hold for $\rmtr _H$. For example, in general we have that
\begin{align*}
\rmtr _H\sqbrac{(A\otimes B)(C\otimes D)}&=\rmtr _H(AC\otimes BD)=\rmtr (AC)BD\ne\rmtr (CA)DB\\
   &=\rmtr _H\sqbrac{(C\otimes D)(A\otimes B)}
\end{align*}
Until now we have studied instruments and observables for the base system. It is sometimes of interest to consider these for the probe system. We define the \textit{post-interaction probe instrument} $\mscripttilde ^\rho\in\rmin (K)$ for all $\rho\in\sscript (H)$, $\sigma\in\sscript (K)$ by
\begin{equation}                % equation (3.5)
\label{eq35}
\mscripttilde _x^\rho (\sigma )=\rmtr _H\sqbrac{(I\otimes F_x)^{1/2}\nu (\rho\otimes\sigma )(I\otimes F_x)^{1/2}}
\end{equation}
We did not define $\mscripttilde _x^\rho (\sigma )$ to be $\rmtr _H\sqbrac{(\nu\otimes\sigma )(I\otimes F_x)}$ as one might expect because the lack of cyclicity of $\rmtr _H$ prevents this latter definition from being self-adjoint. We also define the
\textit{post-interaction probe observable} to be $(\,\mscripttilde ^\rho)^\wedge$. Corresponding to a channel
$\Gamma (\sigma )=\sum S_i\sigma S_i^*$ on $\sscript (K)$ we define the \textit{dual channel} $\Gamma ^*(a)=\sum S_i^*aS_i$ on
$\escript (K)$.

\begin{thm}    % Theorem 3.4
\label{thm33}
Let $\mscript$ be a $C$-nondisturbing $MM$ with $\nu$ given by \eqref{eq31}. For every $\rho\in\sscript (H)$, $\sigma\in\sscript (K)$ we have that
\begin{align}                
\label{eq36}                        % equation (3.6)
\mscripttilde _x^\rho (\sigma )&=\sum _i\elbows{\psi _i,\rho\psi _i}F_x^{1/2}\Gamma _i(\sigma )F_x^{1/2}\\
\label{eq37}                          % equation (3.7)
(\,\mscripttilde ^\rho)_x^\wedge&=\sum _i\elbows{\psi _i,\rho\psi _i}\Gamma _i^*(F_x)
\end{align}
\end{thm}
\begin{proof}
For $\rho\in\sscript (H)$, $\sigma\in\sscript (K)$ we have from \eqref{eq35} that
\begin{align*}
\mscripttilde _x^\rho (\sigma )
   &=\sum _{i,j,k}\elbows{\psi _i,\rho\psi _j}\rmtr _H\paren{\ket{\psi _i}\bra{\psi _j}\otimes F_x^{1/2}B_i^k\sigma B_j^{k*}F_x^{1/2}}\\
   &=\sum _{i,k}\elbows{\psi _i,\rho\psi _i}F_x^{1/2}B_i^k\sigma B_i^{k*}F_x^{1/2}\\
   &=\sum _i\elbows{\psi _i,\rho\psi _i}F_x^{1/2}\Gamma _i(\sigma )F_x^{1/2}
\end{align*}
which verifies \eqref{eq36}. To verify \eqref{eq37}, we see from \eqref{eq36} that
\begin{align*}
\rmtr\sqbrac{\sigma (\,\mscripttilde ^\rho )_x^\wedge}&=\rmtr\sqbrac{\mscripttilde _x^\rho (\sigma )}
   =\sum _i\elbows{\psi _i,\rho\psi _i}\rmtr\sqbrac{\Gamma _i(\sigma )F_x}\\
   &=\sum _i\elbows{\psi _i,\rho\psi _i}\rmtr\sqbrac{\sigma\Gamma _i^*(F_x)}
   =\rmtr\sqbrac{\sigma\sum _i\elbows{\psi _i,\rho\psi _i}\Gamma _i^*(F_x)}
\end{align*}
and \eqref{eq37} follows.
\end{proof}

Since $\brac{\Gamma _i^*(F_x)}$ is an observable for each $i=1,2,\ldots ,n$, it follows from \eqref{eq37} that $(\mscripthat ^\rho )^\wedge$ is a convex combination of these observables depending on $\rho\in\sscript (H)$. This shows that the post-interaction probe observable is obtained by first running the initial probe observable through the channel and then taking a convex combination depending on the input state $\rho$.

The simplest case for the previous theory is when $\nu$ is a unitary $C$-nondisturbing channel so that
$\nu (\rho\otimes\eta )=U(\rho\otimes\eta )U^*$ where $U\in\lscript (H\otimes K)$ is unitary. We then have that
$U(\psi _i\otimes\phi )=\psi _i\otimes V_i\phi$ where $V_i\in\lscript (K)$ are unitary, $i=1,2,\ldots ,n$. We then obtain:
\begin{align}
\nu (\rho\otimes\eta )&=\sum _{i,j}(P_{\psi _i}\rho P_{\psi _j}\times V_i\eta V_j^*)\notag\\                
\label{eq38}                        % equation (3.8)
\mscripthat _x(\rho )&=\sum _{i,j}\rmtr (V_i\eta V_j^*F_x)P_{\psi _i}\rho P_{\psi _j}\\
\label{eq39}                          % equation (3.9)
\mscript _x^{\wedge\wedge}&=\sum _i\rmtr (V_i\eta V_i^*F_x)P_{\psi _i}
\end{align}
Moreover, \eqref{eq36}, \eqref{eq37} become
\begin{align}                
\label{eq310}                        % equation (3.10)
\mscripttilde _x^\rho (\sigma )&=\sum _i\elbows{\psi _i,\rho\psi _i}F_x^{1/2}V_i\sigma V_i^*F_x^{1/2}\\
\label{eq311}                          % equation (3.11)
(\,\mscripttilde ^\rho )_x^\wedge&=\sum _i\elbows{\psi _i,\rho\psi _i}V_i^*F_xV_i
\end{align}
Notice that if $\eta V_i=V_i\eta$, $i=1,2,\ldots ,n$, then we obtain the simple form $\mscript _x^{\wedge\wedge}=\rmtr (\eta F_x)I_H$.

If $\Omega _A$ is a finite set, an \textit{apparatus on} $H\otimes K$ \textit{with outcome-set} $\Omega _A$ is a map $A\colon\sscript (H)\times\Omega\to\escript (K)$ satisfying $x\mapsto A(\rho ,x)$ is an observable on $K$ for every $\rho\in\sscript (H)$ and $\rho\mapsto A(\rho ,x)$ is affine for $x\in\Omega _A$. Thus, we have that
\begin{equation*}
A\paren{\sum\lambda _i\rho _i,x}=\sum\lambda _iA(\rho _i,x)
\end{equation*}
whenever $\lambda _i\ge 0$, $\sum\lambda _i=1$. A \textit{state-dependent} $MM$ is a set of $MM$s of the form
\begin{equation*}
\mscript _\rho =\paren{H,K,\eta ,\nu ,G(\rho ,x)}
\end{equation*}
where $G$ is an apparatus on $H\otimes K$. The motivation behind this definition is the following. Input a state $\rho$ and run the $MM$
$\mscript =(H,K,\eta ,\nu ,F)$ to produce the post-interaction  probe observable which gives the apparatus
$A(\rho ,x)=(\,\mscripthat ^\rho )_x^\wedge$. We then have the state-dependent $MM$ given by
$\mscript _\rho=\paren{H,K,\eta ,\nu ,A(\rho ,x)}$ which is the post-interaction $MM$ resulting from the input state $\rho\in\sscript (H)$.

Now let $\mscript$ be a $C$-nondisturbing $MM$ with $\nu$ given by \eqref{eq31}. Applying \eqref{eq37} the measured apparatus is
\begin{equation*}
A(\rho ,x)=\sum _i\elbows{\psi _i,\rho\psi _i}\Gamma _i^*(F_x)
\end{equation*}
If we now remeasure using the stat-dependent $MM$ $\mscript _\rho$, applying \eqref{eq34} we obtain the apparatus
\begin{align}                % equation (3.12)
\label{eq312}
B(\rho ,x)&=\sum _j\rmtr\sqbrac{\Gamma _j(\eta )\sum _i\elbows{\psi _i,\rho\psi _i}\Gamma _i^*(F_x)}P_{\psi _i}\notag\\
   &=\sum _{i,j}\rmtr\sqbrac{\Gamma _j(\eta )\Gamma _i^*(F_x)}P_{\psi _i}\rho P_{\psi _i}
\end{align}
In the case of a unitary $C$-nondisturbing $MM$, \eqref{eq312} becomes
\begin{equation*}
B(\rho ,x)=\sum _{i,j}\rmtr\sqbrac{V_iV_j\eta (V_iV_j)^*F_x}P_{\psi _i}\rho P_{\psi _i}
\end{equation*}

\section{Two Examples}  % Section 4
This section presents two examples of $C$-nondisturbing $MM$s with unitary channels
$\nu (\rho\otimes P_{\phi _1})=U(\rho\otimes P_{\phi _1})U^*$ where $P_{\phi _1}$ is the initial pure probe state. We write
$\mscript =(H,K,\phi _1,U,F)$.

\begin{exam}{1}  % Example 1
Suppose $\dim H=\dim K=n$ and $\brac{\phi _i}$, $i=1,2,\ldots ,n$, is an orthonormal basis for $K$. Let $V_i$, $i=1,2,\ldots ,n$, be unitary swap operators on $K$ defined by $V_i(\phi _1)=\phi _i$ $V_i(\phi _i)=\phi _1$ and $V_i(\phi _j)=\phi _j$ for $j\ne i$. Defining the unitary operator
$U\in\lscript (H\otimes K)$ by $U(\psi _i\otimes\phi )=\psi _i\otimes V_i(\phi )$ for all $\phi\in K$, we have that $U$ is $C$-nondisturbing. Letting $\nu$ be the unitary channel on $H\otimes K$ given by $\nu (\sigma )=U\sigma U^*$ we have that $\nu$ is $C$-nondisturbing. Form the
$C$-nondisturbing $MM$ $\mscript =(H,K,\phi _1,U,F)$. Applying \eqref{eq37}, the measured instrument is
\begin{align}                % equation (4.1)
\label{eq41}
\mscripthat _x(\rho )&=\sum _{i,j}\rmtr (V_iP_{\phi _1}V_j^*F_x)P_{\psi _i}\rho P_{\psi _j}
   =\sum _{i,j}\rmtr\paren{\ket{V_i\phi _1}\bra{V_j\phi _1}F_x}P_{\psi _i}\rho P_{\psi _j}\notag\\
   &=\sum _{i,j}\rmtr\paren{\ket{\phi _i}\bra{\phi _j}F_x}P_{\psi _i}\rho P_{\psi _j}=\sum _{i,j}\elbows{\phi _j,F_x\phi _i}P_{\psi _i}\rho P_{\psi _j}
\end{align}
By \eqref{eq38} the measured observable is
\begin{align*}
\mscript _x^{\wedge\wedge}&=\sum _i\rmtr (V_iP_{\phi _1}V_i^*F_x)P_{\psi _i}
   =\sum _i\rmtr\paren{\ket{V_i\phi _1}\bra{V_i\phi _1}F_x}P_{\psi _i}\\
   &=\sum _i\rmtr\paren{\ket{\phi _i}\bra{\phi _i}F_x}P_{\psi _i}=\sum _i\elbows{\phi _i,F_x\phi _i}P_{\psi _i}
\end{align*}
By \eqref{eq31} the channel becomes
\begin{align}                % equation (4.2)
\label{eq42}
\nu (\rho\otimes P_{\phi _1})&=\sum _{i,j}(P_{\psi _i}\rho P_{\psi _j}\otimes V_iP_{\phi _1}V_j^*)
  =\sum _{i,j}\paren{P_{\psi _i}\rho P_{\psi _j}\otimes\ket{V_i\phi _1}\ket{V_j\phi _1}}\notag\\
  &=\sum _{i,j}\paren{P_{\psi _i}\rho P_{\psi _j}\otimes\ket{\phi _i}\bra{\phi _j}}
\end{align}
If $\rho$ is $C$-measurable so that $\rho\psi _i=\lambda _i\psi _i$, $i=1,2,\ldots ,n$, then \eqref{eq41} and \eqref{eq42} have the simple forms
\begin{align*}
\mscripthat _x(\rho )&=\sum _i\lambda _i\elbows{\phi _i,F_x\phi _i}P_{\psi _i}\\
\nu (\rho\otimes P_{\phi _1})&=\sum (\lambda _iP_{\psi _i}\otimes P_{\phi _i})\qed
\end{align*}
\end{exam}
\medskip

\begin{exam}{2}  % Example 2
Let $\dim H=n$, $\dim K=m$ and let $\brac{\phi _i}$ be an orthonormal basis for $K$. Let $i=\sqrt{-1\,}$ be in imaginary unit and define the unitary operators
\begin{equation*}
V_j(\phi _r)=\frac{1}{\sqrt{m\,}}\,\sum _{s=1}^me^{2\pi ijrs/m}\phi _s
\end{equation*}
for $j=1,2,\ldots ,n$, $r,s=1,2,\ldots ,m$ \cite{hz12}. Define the unitary operator $U=\sum _{j=1}^n(P_{\psi _j}\otimes V_j)$ so that
\begin{equation*}
U=(\psi _j\otimes\phi _r)=\psi _j\otimes U_j\phi _r=\psi _j\otimes\frac{1}{\sqrt{m\,}}\,\sum _{s=1}^me^{2\pi ijrs/m}\phi _s
\end{equation*}
Then $U$ and the unitary channel $\nu (\sigma )=U\sigma U^*$ and $C$-nondisturbing. Form the $C$-nondisturbing $MM$
$\mscript =(H,K,\phi _1,U,F)$. As in \eqref{eq42} the channel becomes
\begin{align*}
\nu (\rho\otimes P_{\phi _1})&=\sum _{j,k}\paren{P_{\psi _j}\rho P_{\psi _k}\otimes\ket{V_j\phi _1}\bra{V_k\phi _1}}\\
   &=\tfrac{1}{m}\,\sum _{j,k}
   \paren{P_{\psi _j}\rho P_{\psi _k}\otimes\ket{\sum _{s=1}^me^{2\pi ijs/m}\phi _s}\bra{\sum _{t=1}^me^{2\pi ikt/m}\phi _t}}\\
   &=\tfrac{1}{m}\sum _{\substack{j,k\\s,t}}e^{2\pi i(js-kt)/m}\paren{P_{\psi _j}\rho\psi _k\otimes\ket{\phi _s}\bra{\phi _t}}
\end{align*}
In a similar way, we obtain
\begin{equation}                % equation (4,3)
\label{eq43}
\rmtr\paren{\ket{V_j\phi _1}\bra{V_k\phi _1}F_x}=\tfrac{1}{m}\,\sum _{s,t}e^{2\pi i(js-kt)/m}\elbows{\phi _t,F_x\phi _s}
\end{equation}
and applying \eqref{eq38}, the measured instrument becomes
\begin{align*}
\mscripthat _x(\rho )&=\sum _{j,k}\rmtr\paren{\ket{V_j\phi _1}\bra{V_k\phi _1}F_x}P_{\psi _j}\rho P_{\psi _k}\\
  &=\tfrac{1}{m}\,\sum _{\substack{j,k\\s,t}}e^{2\pi i(js-kt)/m}\elbows{\phi _t,F_x\phi _s}P_{\psi _j}\rho P_{\psi _k}
\end{align*}
Applying \eqref{eq39} and \eqref{eq43}, the measured observable becomes
\begin{equation}                % equation (4.4)
\label{eq44}
\mscript _x^{\wedge\wedge}=\tfrac{1}{m}\,\sum _{j,s,t}e^{2\pi ij(s-t)/m}\elbows{\phi _t,F_x\phi _x}P_{\psi _j}
\end{equation}

As a simple example, suppose $\phi _k$ are eigenvectors of $F_x$ for all $x\in\Omega _F$. Then $F_x\phi _s=c_x^s$ for all $x,s$, where
$0\le c_x^s\le 1$. Then by \eqref{eq44} we obtain
\begin{align*}
\mscript _x^{\wedge\wedge}&=\tfrac{1}{m}\,\sum _{j,s,t}e^{2\pi ij(s-t)/m}\elbows{\phi _t,c_x^s\phi _s}P_{\psi _j}
  =\tfrac{1}{m}\,\sum _{j,s}c_x^sP_{\psi _j}\\
  &=\paren{\tfrac{1}{m}\,\sum _sc_x^s}I_H=\elbows{F_x}I_H
\end{align*}
where $\elbows{F_x}$ is the average eigenvalue of $F_x$.\hfill\qedsymbol 
\end{exam}

\end{document}